\newif\ifsoda
\newif\ifarxiv

\sodafalse
\arxivfalse

\arxivtrue

\ifsoda 
\RequirePackage[loading]{tracefnt}
\documentclass[twoside,leqno]{article}
\usepackage[letterpaper]{geometry}
\usepackage{ltexpprt}
\usepackage{hyperref}
\usepackage{amsmath,amssymb}        
\usepackage{dsfont}         
\usepackage{graphicx}               
\usepackage{algorithm}
\usepackage{algpseudocode}
\usepackage{enumitem}               
\usepackage{xcolor}
\fi

\ifarxiv
\documentclass[10pt]{article}
\usepackage[letterpaper,margin=1in]{geometry}
\usepackage[colorlinks=true, allcolors=blue]{hyperref}
\usepackage[utf8]{inputenc} 
\usepackage{url}            
\usepackage{booktabs}       
\usepackage{amsfonts}       
\usepackage{nicefrac}       
\usepackage{microtype}      
\usepackage{algorithm}
\usepackage{algpseudocode}
\usepackage{amsmath}
\usepackage{dsfont}
\usepackage{cases}          
\usepackage{graphicx}       
\usepackage[english]{babel}

\usepackage{amsthm}
\usepackage{amsmath,amssymb}
\usepackage{comment}        
\usepackage{enumitem}  
\usepackage{parskip}

\usepackage[numbers]{natbib}

\newtheorem{theorem}{Theorem}[section]
\newtheorem{proposition}[theorem]{Proposition}
\newtheorem{lemma}[theorem]{Lemma}

\newtheorem{Definition}[theorem]{Definition}

\newtheorem*{definition*}{Definition}

\newtheorem*{theorem*}{Theorem}

\newtheorem*{lemma*}{Lemma}

\usepackage{titlesec}
\titleformat{\section}{\large\bfseries}{\thesection}{1em}{}
\titleformat{\subsubsection}[runin]
       {\normalfont\bfseries}
       {\thesubsection}
       {0.5em}
       {}
       []
\titleformat{\subsection}[runin]
       {\normalfont\bfseries}
       {\thesubsection}
       {0.5em}
       {}
       []

\fi

\newcommand{\Acal}{\mathcal{A}}
\newcommand{\Bcal}{\mathcal{B}}

\newcommand{\Ocal}{\mathcal{O}}

\newcommand{\Scal}{\mathcal{S}}
\newcommand{\Tcal}{\mathcal{T}}

\newcommand{\Vcal}{\mathcal{V}}

\newcommand{\Nbb}{\mathbb{N}}

\newcommand{\Zbb}{\mathbb{Z}}

\newcommand{\bx}{{\boldsymbol x}}

\usepackage{bbm}

\newcommand{\1}{\mathds{1}}

\newcommand{\ceil}[1]{{\left\lceil #1 \right\rceil}}
\newcommand{\floor}[1]{{\lfloor #1 \rfloor}}

\newcommand{\cost}{\normalfont\text{Cost}}

\newcommand{\cte}{\text{CTE}}

\newcommand{\acte}{\text{ACTE}}
\newcommand{\tm}{\text{TM}}
\renewcommand{\tau}{i}

\ifarxiv
\title{\Large Breaking the $k/\log k$ Barrier in Collective Tree Exploration via Tree-Mining}
\author{Romain Cosson \\ {\small \texttt{romain.cosson@inria.fr}}}
\date{}
\fi

\begin{document}

\maketitle
\begin{abstract}
In collective tree exploration, a team of $k$ mobile agents is tasked to go through all edges of an unknown tree as fast as possible.
An edge of the tree is revealed to the team when one agent becomes adjacent to that edge.
The agents start from the root and all move synchronously along one adjacent edge in each round.  
Communication between the agents is unrestricted, and they are, therefore, centrally controlled by a single exploration algorithm. 
The algorithm's guarantee is typically compared to the number of rounds required by the agents to go through all edges if they had known the tree in advance. 
This quantity is at least $\max\{2n/k,2D\}$ where $n$ is the number of nodes and $D$ is the tree depth.
Since the introduction of the problem by \cite{fraigniaud2004collective}, two types of guarantees have emerged: the first takes the form $r(k)(n/k+D)$, where $r(k)$ is called the competitive ratio, and the other takes the form $2n/k+f(k,D)$, where $f(k,D)$ is called the competitive overhead. 
In this paper, we present the first algorithm with linear-in-$D$ competitive overhead, thereby reconciling both approaches. 
Specifically, our bound is in $2n/k + \Ocal(k^{\log_2(k)-1} D)$ and leads to a competitive ratio in $\Ocal(k/\exp(\sqrt{\ln 2\ln k}))$. 
This is the first improvement over $\Ocal(k/\ln k)$ since the introduction of the problem, twenty years ago. 
Our algorithm is developed for an asynchronous generalization of collective tree exploration (ACTE). 
It belongs to a broad class of \textit{locally-greedy} exploration algorithms that we define.
We show that the analysis of locally-greedy algorithms can be seen through the lens of a 2-player game that we call the \textit{tree-mining} game and which could be of independent interest. 
\end{abstract}
\section{Introduction}
Exploration problems on unknown graphs or other geometric spaces have been studied in several fields, e.g., theoretical computer science, computational geometry, robotics and artificial intelligence. The present study concerns the problem of collective tree exploration, introduced in \cite{fraigniaud2004collective}. The goal is for a team of agents or robots, initially located at the root of an unknown tree, to go through all of its edges and then return to the root, as fast as possible. At each round, all robots move synchronously along an adjacent edge to reach a neighboring node. When a robot attains a new node, its adjacent edges are revealed to the team. Following the complete communication model of \cite{fraigniaud2004collective}, we assume that robots can communicate and compute at no cost. The team thus shares at all time a map of the sub-tree that has already been explored and the agents are controlled centrally by a single algorithm.

\subsection{Main result} In this paper, we present an algorithm that achieves collective tree exploration with $k$ robots in $2n/k+\Ocal\left(k^{\log_2(k)-1}D\right)$ synchronous rounds for any tree with $n$ nodes and depth $D$. This algorithm induces a new competitive ratio for collective tree exploration of order $k/\exp(\sqrt{\ln 2\ln k})$, improving over the order $k/\ln k$ competitive ratio of \cite{fraigniaud2004collective}.

Our algorithm is based on a strategy for a two-player game, the \textit{tree-mining game}, that we introduce and analyze. In this game, the player attempts to lead $k$ miners deeper in a tree-mine, where leaves of the tree represent digging positions, while the adversary attempts to hinder their progression. An $f(k,D)$-bounded strategy of the player is one such that the miners are guaranteed to all reach depth $D$ in less than $f(k,D)$ moves. We present such a strategy with $f(k,D)= \Ocal\left(k^{\log_2 k}D\right)$.

The tree-mining game is an abstraction designed to study the competitive overhead of a broad class of exploration algorithms that we call \textit{locally-greedy exploration algorithms}. Though most exploration algorithms described in prior works are effectively locally-greedy, the idea that they form a family of algorithms that can be analyzed from a simpler viewpoint is new. Also, we introduce the setting of \textit{asynchronous collective tree exploration} (ACTE), drawing inspiration from the adversarial setting of \cite{cosson2023efficient} while restraining robot perception capabilities. The guarantees we obtain are initially derived for this generalization of collective tree exploration. 

\subsection{Competitive analysis of collective tree exploration} Competitive analysis is a framework that studies the performance of \textit{online} algorithms relative to their optimal offline counterpart. The origin of this framework dates back to the work of \cite{sleator1985amortized} on the \textit{Move-To-Front} heuristic for the list update problem.
The term `competitive analysis' was coined shortly after by
\cite{karlin1988competitive} in the study of a caching problem. Competitive analysis quickly became the standard tool for rigorous analysis of online algorithms, with many celebrated problems such as the $k$-server problem (see \cite{koutsoupias2009k} for a review), layered graph traversal \cite{papadimitriou1991shortest},
online bipartite matching \cite{goel2008online,karande2011online}, among others. 

The framework of competitive analysis naturally extends to \textit{search} problems, where the time to find a lost ``treasure'' is typically compared to the minimum time required to reach it if its position was known to the searcher. Perhaps the most notable example is the \textit{cow-path problem} \cite{kao1996searching}, also called \textit{linear search problem} \cite{baezayates1993searching}, which traces its origin back to Bellman \cite{bellman1963optimal}. In its simplest version, a farmer located on a 1-dimensional field is searching for its cow, located at some coordinate $x\in \Zbb$. The farmer moves a constant speed. If the farmer knew the position of the cow, it could attain it in only $|x|$ steps. We thus say that the search strategy of the farmer is $c$-competitive if it is guaranteed to find its cow in at most $c|x|$ steps, and we call $c$ the competitive ratio. For this problem, the simple ``doubling strategy'' is $9$-competitive and is optimal among deterministic algorithms \cite{baezayates1993searching}. The competitive ratio can be improved to $4.5911$ with a randomized strategy \cite{karande2011online}. Other search problems have received a competitive analysis, such as \cite{feinerman2012collaborative}.

In this spirit, \cite{fraigniaud2004collective,fraigniaud2006collective} introduced an \textit{exploration} problem called collective tree exploration (CTE). A team $k\in \Nbb$ agents is tasked to go through all the edges of a tree $T$ and return to the origin. The \textit{online} problem can be thought as the situation where the tree is initially unknown to the agents whereas its \textit{offline} counterpart, also known as the $k$-travelling salesmen problem on $T$ \cite{averbakh1996heuristic,averbakh1997p}, corresponds to the situation where the agents are provided a map of $T$ beforehand. 
Denoting by $\texttt{Runtime}(\Acal,k,T)$ the number of rounds required by $k$ agents to explore a tree $T$ with an online exploration algorithm $\Acal$, and denoting by $\texttt{OPT}(k,T)$ its optimal offline counterpart, the competitive ratio of algorithm $\Acal$ is defined as, 
\begin{equation*}
    \texttt{Competitive Ratio}(\Acal,k) = \max_{T\in \Tcal}\frac{\texttt{Runtime}(\Acal,k,T)}{\texttt{OPT}(k,T)}.
\end{equation*}
Though $\texttt{OPT}(k,T)$ can be NP-hard to compute \cite{fraigniaud2006collective}, it is clearly greater than $\max\{2n/k,2D\}$ where $n$ is the number of nodes and $D$ is the depth of the tree \cite{brass2011multirobot}, thus $\texttt{OPT}(k,T)\geq \frac{n}{k}+D$. This bound is in fact tight up to a factor $2$ \cite{dynia2006power, OrtolfS14}. Consequently, we say that an exploration algorithm is $r(k)$-competitive when its runtime is bounded by $r(k)(\frac{n}{k}+D)$ on any tree with $n$ nodes and depth $D$. 

As a simple example, the algorithm leaving $k-1$ robots idle at the origin and using one robot to perform a depth-first search runs in exactly $2(n-1)$ rounds and therefore has a $\Ocal(k)$ competitive ratio. As they introduced the problem, \cite{fraigniaud2004collective} proposed a simple algorithm (later qualified of `greedy' by \cite{higashikawa2014online}) with a guarantee in $\Ocal({n}/{\ln k}+D)$ and thus achieving a $\Ocal({k}/{\ln k})$ competitive ratio, thus providing a logarithmic improvement over depth-first search. In this paper, we present an algorithm with a competitive ratio in $\Ocal(k/\exp(\sqrt{\ln 2\ln k}))$, thus providing an improvement exceeding any poly-logarithmic function of $k$, but which remains far from the best lower-bond on the competitive ratio, which is in $\Omega({\ln k}/{\ln \ln k})$ \cite{dynia2007robots}.

Other types of guarantees for collective tree exploration have been proposed in the literature (see \cite{OrtolfS14, dynia2006smart, brass2011multirobot}, and references therein) but they all have a super-linear dependence in $(n,D)$ and are thus unable to improve the competitive ratio as a function of $k$ only. Interestingly, the work of \cite{brass2011multirobot} proposed a novel competitive viewpoint on collective tree exploration. Their analysis of the algorithm of \cite{fraigniaud2006collective} yields a guarantee in $\frac{2n}{k}+\Ocal((k+D)^k)$, hence with an optimal dependence in $n$ and with an additive cost which does not depend on $n$. This result poses the question of the (additive) \textit{competitive overhead} of an exploration algorithm $\Acal$, which we can formally define as, 
\begin{equation*}
    \texttt{Competitive Overhead}(\Acal,k,D) = \max_{T\in \Tcal(D)} \texttt{Runtime}(\Acal,k,T)-\texttt{OPT}(k,T),
\end{equation*}
where $\Tcal(D)$ denotes the set of all trees with depth $D$. We note that the notion of competitive overhead is immediately related to the notion of \textit{penalty} appearing in the pioneering work of \cite{panaite1999exploring} on graph exploration with a single mobile agent. A $\frac{2n}{k}+\Ocal(D^2\ln k)$ exploration algorithm was later proposed by \cite{cosson2023efficient}, improving over the result of \cite{brass2011multirobot} for all values of $k$ and $D$. 

We summarize the above discussion in Table~\ref{table:1} below. 

\renewcommand{\arraystretch}{1.4}
\setlength{\tabcolsep}{20pt}

\begin{table}[h!]
\centering
\begin{tabular}{||c | c c||} 
 \hline
  $\Ocal(\cdot)$& Competitive Ratio & Competitive Overhead \\
  Runtime  & $r(k)(\frac{n}{k}+D)$ & $\frac{2n}{k}+f(k,D)$ \\[0.5ex] 
 \hline\hline
 \cite{fraigniaud2006collective} & $k/\ln k$ & - \\
 \cite{cosson2023efficient} & - & $\ln kD^2$\\
 This work & $k/\exp(\sqrt{\ln 2\ln k})$ & $k^{\log_2(k)-1}D$\\ [1ex] 
 \hline
\end{tabular}
\caption{Known results on the competitive ratio and the competitive overhead of collective tree exploration. Results are given up to a multiplicative constant.}
\label{table:1}
\end{table}

\subsection{Notations and definitions} In what follows, $\ln(\cdot)$ refers to the natural logarithm and $\log_2(\cdot)$ to the logarithm in base $2$. For an integer $k$ we use the abbreviation $[k] = \{1,\dots,k\}$.

A tree $T=(V,E)$ is defined by its set of nodes $V$ and edges $E \subset V\times V$, and some $\texttt{root}\in V$. A partially explored tree $T=(V,E)$ is a tree where some edges may have a single discovered endpoint. This data structure represents the knowledge a team of explorers in the course of their exploration. 
A (synchronous) \textit{collective tree exploration algorithm} $\Acal$ is a function that, for each explorer in the partially explored tree, selects an adjacent edge that it will traverse at that round. Robots are allowed to stay at their current position (self-loops). The algorithm's performance is evaluated by $\texttt{Runtime}(\Acal,k,T)$ defined as the number of rounds required for $k$ robots initially located at the root, to go through all edges of the tree $T$, and return to the root.

\subsection{Preliminaries} This work provides the first collective exploration algorithm with linear-in-$D$ competitive overhead.

\begin{theorem}[Competitive Overhead]\label{th:intro}
    There exists a collective tree exploration algorithm explorers satisfying for any tree $T$ with $n$ nodes and depth $D$,
    $${\normalfont\texttt{Runtime}}(\Acal,k,T)\leq \frac{2n}{k}+\Ocal(k^{\log_2(k)-1}D).$$  
\end{theorem}
This result also improves the competitive ratio of collective tree exploration. Surprisingly, for that purpose, we shall only use a fraction of the robots. 
\begin{lemma}\label{lemma:reduction}
    For any collective tree exploration algorithm $\Acal$ with competitive overhead in $\Ocal\left(\frac{f(k)}{k}D\right)$, where $f$ is some increasing real-valued function, there exists a collective tree exploration algorithm $\Acal'$ with competitive ratio in $\Ocal\left(\frac{k}{f^{-1}(k)}\right)$, where $f^{-1}$ denotes the inverse of $f$. $\Acal'$ is simply defined as the algorithm using $\Acal$ with a fraction $k' = \floor{f^{-1}(k)}$ of the robots, while maintaining all other robots idle at the root. 
\end{lemma}
\begin{proof}
    The runtime of the algorithm $\Acal'$ is bounded by $\frac{2n}{k'}+\Ocal\left(\frac{f(k')}{k'}D\right)$. Since we have $f(k')\leq k$, we get, 
    $$\texttt{Runtime}(\Acal',k,T)\leq \Ocal\left(\frac{k}{k'}\right)\left(\frac{n}{k}+D\right).$$
    Observing that $k'\geq f^{-1}(k)-1\geq 1 $ allows to conclude.
\end{proof}
We then apply Lemma \ref{lemma:reduction} and Theorem \ref{th:intro} to $f(k) = k^{\log_2(k)}$ to get the following result. Note that the quantity $\exp\left(\sqrt{\ln 2 \ln k}\right) = k^{\frac{1}{\sqrt{\log_2(k)}}}$, which appears below, is asymptotically greater than any poly-logarithmic function of $k$, but smaller than any polynomial of $k$.
\begin{theorem}[Competitive Ratio]
There exists a $\mathcal{O}\left(k/\exp\left(\sqrt{\ln 2 \ln k}\right)\right)$-competitive collective tree exploration algorithm.
\end{theorem}
The algorithm $\Acal$ of Theorem \ref{th:intro} belongs to a broad class of algorithms that we call \textit{locally-greedy exploration algorithms}. This class relaxes the notion of greedy exploration algorithms of \cite{higashikawa2014online}. We shall see in Section \ref{sec: red} that the competitive overhead of locally-greedy exploration algorithm can be tightly analyzed through the tree-mining game, studied in Section \ref{sec: game}.
\begin{Definition}[Locally-Greedy CTE] A locally-greedy collective tree exploration algorithm is such that at any synchronous round, if a node is adjacent to $e$ unexplored edges and is populated by $x$ robots, then at the next round, this node will be adjacent to exactly $\max\{0,e-x\}$ unexplored edges. 
\end{Definition}
Another way to define locally-greedy algorithms is to assume that robots move sequentially rather than synchronously. In this context, a locally-greedy algorithm is one in which the moving robot always selects an unexplored edge when one is incident. We call this setting \textit{asynchronous collective tree exploration} (ACTE) because we let an adversary choose at each round the robot which is allowed to perform move. In this setting, a locally-greedy algorithm is one such that a moving robot always prefers a to traverse an unexplored edge if possible. The setting of asynchronous collective tree exploration is inspired from the adversarial setting of \cite{cosson2023efficient}, with an additional limitation on the perception capabilities of the robots. Note that it is defined here in the complete communication model, where the agents are controlled centrally by one algorithm. An asynchronous version of collective tree exploration remains to be defined in distributed variants of collective tree exploration, such as the `write-read' communication model of \cite{fraigniaud2006collective}. 

\subsection{Paper outline} In Section \ref{sec: game}, we study the tree-mining game and present a $\Ocal(k^{\log_2 k})$-bounded strategy for the player of the game. In Section \ref{sec: red} we describe the reductions connecting the tree-mining game to the analysis of the overhead of locally-greedy algorithms for synchronous and asynchronous collective tree exploration. Finally, Appendix~\ref{sec: calculus} and Appendix \ref{sec:details} contain important technical details to support the tree-mining strategy defined in Section \ref{sec: game}.

\section{The tree-mining game}\label{sec: game}
\subsection{Game description}\label{sec: game-description}
We now describe the \textit{tree-mining game}. As we will see in Section \ref{sec: red}, this game provides an analysis of locally-greedy algorithms for collective tree exploration.

\subsubsection*{The board} Let $k\geq 2$ be some integer. The board of the game is defined by a pair $\Tcal = (T, \bx)$ where $T=(V,E,L)$ is a rooted tree with nodes 
$V \subset \Vcal$ edges $E\subset V\times V$ and a set of \textit{active leaves} $L \subset V$, and where 
$\bx = (x_\ell)_{\ell\in L}\in\Nbb^L$ is a configuration representing the number of miners located at active leaves of $T$, for a total of $k$ miners, i.e. $\sum_{\ell\in L} x_\ell=k$. There is at least one miner per active leaf. The set of all possible boards is denoted by $\mathfrak{T}$. The board at round $\tau \in \Nbb$ is denoted by $\Tcal(\tau) = (T(\tau),\bx(\tau))$. The game starts with $T(0)$ reduced to a single active leaf, on which all $k$ miners are located.

\subsubsection*{Adversary} At the beginning of round $\tau$, the adversary observes $\Tcal(\tau)$ and is the first to play. It chooses an active leaf $\ell(\tau)\in L(\tau)$ which is deactivated and provides it with a number of active children $c(\tau)\in \Nbb$, which has to be less or equal to $x_{\ell(\tau)}(\tau)-1$. This choice entirely determines $T(\tau+1)$. The adversary thus controls the evolution of the tree structure with a strategy of the form,
\begin{align*}
s_a: \quad \quad \mathfrak{T} &\longrightarrow V\times\Nbb\\
 \Tcal&\longmapsto (\ell,c).
\end{align*}

\subsubsection*{Player} The move of the adversary is observed by the player, along with the state of the board. The player must then relocate the $x_{\ell(\tau)}(\tau)$ miners on the deactivated leaf. The player is required to send at least one miner to each of the $c(\tau)$ new children in $\ell(\tau)$, and she is free to choose any destination of $L(\tau+1)$ for the remaining miners. This choice entirely determines $\bx(\tau+1)$. The player thus controls the evolution of the configuration with a strategy of the form, 
\begin{align*}
s_p: \mathfrak{T} \times V\times \Nbb &\longrightarrow \mathfrak{T}\\
 (\Tcal,\ell,c)&\longmapsto \Tcal'.
\end{align*}
If the adversary deactivates all active nodes, providing them with no child, the game is considered finished.

\subsubsection*{Cost}
At each round $t$, the cost of the game is updated. The cost counts the total number of moves that were performed by the miners. A discount of $2$ can be given for any new edge created by the adversary, but it will be neglected in most of the analysis\footnote{The discount accounts for the fact that the creation of a leaf with a single miner and its subsequent deletion should not affect the overall cost, since the corresponding edge is traversed exactly twice. It will only be used for more precise results with $k\in \{2,3\}$.}.
\begin{align*}
\cost(\tau+1) = \cost(\tau) + d(\bx(\tau),\bx(\tau+1))-2c(\tau),
\end{align*}
where $d(\bx,\bx')$ denotes the transport distance in the tree between two configurations $\bx$ and $\bx'$. The goal of the player is to have the miners go deeper in the mine while enduring a limited cost.

\subsubsection*{Bounded strategies} We denote by $\Scal_p$ (resp. $\Scal_a$) the set of all valid strategies for the player (resp. the adversary). For a pair $(s_p,s_a) \in \Scal_p\times \Scal_a$, and some depth $D\in \Nbb$, we define $\cost(s_p,s_a,k,D)$ as the maximum cost attained while some miner is at depth less or equal to $D$, if the player (resp. the adversary) uses $s_p$ (resp. $s_a$). For a strategy of the player $s_p\in \Scal_p$, we then define $f_{s_p}(k,D)$ as follows, 
\begin{align*}
f_{s_p}(k,D) &= \max_{s_a \in \Scal_a} \cost(s_p,s_a,k,D)
\end{align*}
For a bi-variate function $f(\cdot,\cdot)$, we say that $s_p$ is $f(k,D)$-bounded if we have $\forall k,D: f_{s_p}(k,D)\leq f(k,D)$. In Figure \ref{fig:first_illustration}, we provide an illustration of a few rounds of the tree-mining game with $k=3$. The interest of the tree-mining game lies in its close connection to collective exploration, highlighted by the following result.
\begin{theorem}[from Section \ref{sec: red}]\label{th: reduction} Any $f(k,D)$-bounded strategy for the tree-mining game induces a collective tree exploration algorithm $\Acal$ satisfying for any tree $T$ with $n$ nodes and depth $D$,  
$${\normalfont\texttt{Runtime}}(\Acal,k,T) \leq \ceil{\frac{2n}{k}+\frac{f(k,D)}{k}} + D.$$
\end{theorem}

We will spend the rest of this section constructing a $\Ocal(k^{\log_2 k}D)$-bounded strategy for the tree-mining game.

\begin{figure}[h!]
\centering
\includegraphics[trim={0 4cm 0 0.5cm},width=\textwidth]{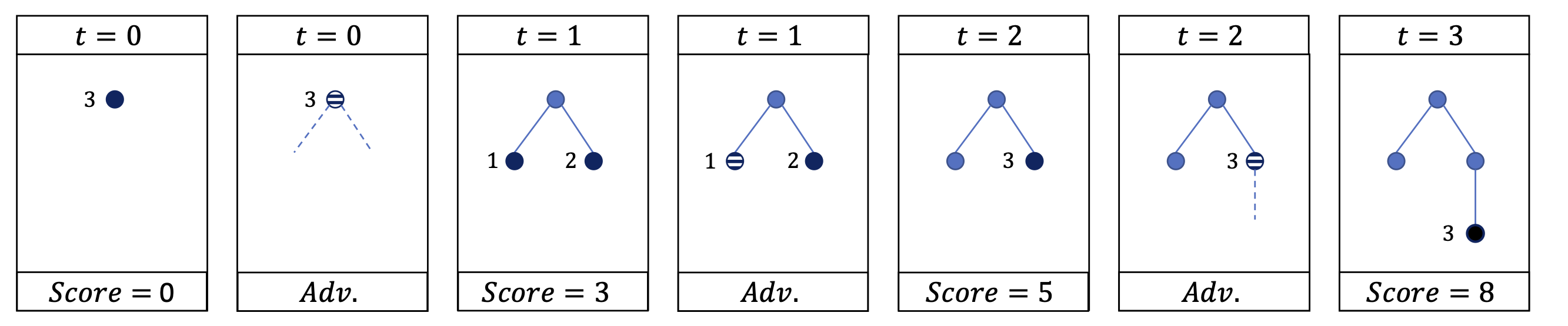}
\caption{Example of a game with $k=3$. The cost is tracked through time. The abbreviation \textit{Adv.} stands for ``Adversary'' and corresponds to the representation of the instant right after the decision of the adversary. Active leaves are represented in dark blue while other nodes are in light blue. }
\label{fig:first_illustration}
\end{figure}

\subsection{Tree-mining with two miners}\label{sec:ksmall} For $k=2$, it is clear that the board of the game at round $\tau$ corresponds to a line-tree with $\tau$ edges and depth $\tau$. The leaf of this line-tree contains two miners. The adversary has no other choice than select this leaf and provide it with one child, and the player must send both miners to this new child. We note that the cost never increases and we thus state,
\begin{proposition}\label{prop: kequals2} For $k=2$, there is a unique strategy available to the player, and it satisfies $f_{s_p}(2,D)=0$.
\end{proposition}
Note that in light of Theorem \ref{th:mainreduction}, this readily recovers the collective tree exploration algorithm of \cite{brass2011multirobot} for $k=2$ robots, which runs in $n+D$ synchronous steps. 

\subsection{Tree-mining with three miners} The case $k=3$ is a variant of the aforementioned cow-path problem, for which the doubling strategy is optimal (see \cite{kao1996searching} and references therein). Observe that at any round, there are at most two active leaves in the board. The tree defined by their ancestors is called the \textit{active sub-tree}. This tree can be represented by a triple $(d,\delta_1,\delta_2)$ where $d$ is the depth of the lowest common ancestor of the active leaves, $d+\delta_1$ is the depth of the alone miner and $d+\delta_2$ is the depth of the other two miners (if all three miners are on the same node, we let $\delta_1=\delta_2=0$). We consider the following strategy for the player, after the adversary's move. An illustration of the strategy is provided in Figure \ref{fig:keq3}.
    \begin{itemize}
        \item[--] \textbf{if} the adversary chooses the node with two miners and provides it with one child \textbf{then}
        \begin{itemize}
            \item \textbf{if} $\delta_2 < 2\delta_1-1$ \textbf{then} the player assigns both miners to go to the new leaf,
            \item \textbf{if} $\delta_2= 2\delta_1-1$ \textbf{then} the player sends one of both miners to the other leaf at distance $\delta_1+\delta_2$,
        \end{itemize}
        \item[--] \textbf{else} the player performs the only possible move. 
    \end{itemize}
\begin{figure}[h!]
\centering
\includegraphics[trim={0 4cm 0 1.5cm},width=\textwidth]{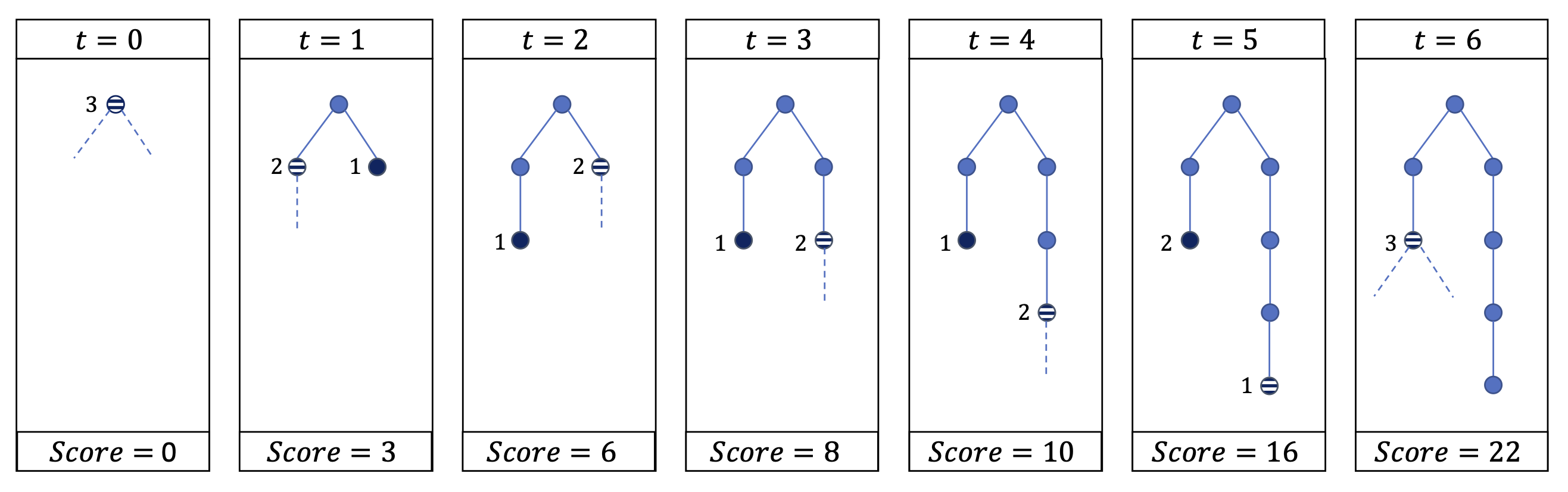}
\caption{Example of a game with $3$ miners where the player uses the strategy described above. The board is represented once per round $\tau\leq 6$, after the adversary's move which is represented as dashed. Active leaves are in dark blue.}
\label{fig:keq3}
\end{figure}
We now provide an analysis, inspired from dynamic programming, which will shed light on the more general case for which $k\geq 4$. In the strategy above, it appears clearly that if at some round the active sub-tree is of the form $(d,\delta, \delta)$, with $\delta\geq 1$, then it attains in a finite number of rounds an active sub-tree that is either of the form $(d,2\delta,2\delta)$ or of the form $(d',0,0)$ with $d+\delta\leq d'\leq d+2\delta$. 
In the first case, the cost increases by at most $3\delta$. In the other case, the cost increases by at most $11\delta\leq 11(d'-d)$. We call such a sequence of rounds an \textit{epoch} and we observe that the algorithm works by iterating epochs indefinitely. We also call by epoch the single round that goes from an active sub-tree of the form $(d,0,0)$ to an active sub-tree of the form $(d+1,0,0)$ or $(d,1,1)$, with an associated cost of at most $1$. We now derive the following lemma. 

\begin{lemma}\label{lemma:keq3} For any strategy of the adversary, if an epoch ends with a board of the form $(d,\delta,\delta)$ then the corresponding cost of the tree mining game is at most $11d+3(d+\delta)$.
\end{lemma}
\begin{proof}
    The proof works by induction. The result is clear at initialisation, with an active sub-tree of the form $(0,0,0)$. Then consider an epoch ending with a tree of the form $(d',\delta',\delta')$ and write $(d,\delta, \delta)$ the form of the active sub-tree when the epoch started, and for which we assume that the property is true. Observe in all of the cases above, the cost incurred during the epoch is bounded by $3(d'+\delta'-d-\delta) + 11(d'-d)$. Thus since the cost at the start of the epoch is bounded by $11d+3(d+\delta)$, it is bounded by $11d'+3(d'+\delta')$ at the end of the epoch.
\end{proof}
\begin{proposition}\label{prop:keq3}
The strategy $s_p$ defined above satisfies $f_{s_p}(3,D)\leq 14D$.
\end{proposition}

\begin{proof}
We consider a board in which some miner is at depth $D'\leq D$. Note the board may well be in the middle of an epoch, so we can't immediately apply the preceding lemma. Instead we imagine that the adversary kills the other node, thus finishing the epoch with an active sub-tree of the form $(D',0,0)$, to which we can apply Lemma \ref{lemma:keq3} to obtain a cost of at most $3D+11D = 14D$.
\end{proof}
Note that in light of Theorem \ref{th:mainreduction}, this result induces the first $\frac{2n}{3}+O(D)$ exploration algorithm with $k=3$ robots, thereby partially answering a previous open question of \cite{higashikawa2014online}. In the sequel, we generalize the idea of looking at the shape of the active sub-tree by introducing the following definition.

\begin{Definition}
For integers $d\leq D$, the board $\Tcal$ has a $(D,d)$-structure if it satisfies the following conditions,
\begin{itemize}
    \item[--] the highest active leaf is at depth $D$,
    \item[--] the lowest common ancestor of all active leaves is at depth $d$.
\end{itemize}
\end{Definition}
\subsection{Tree-mining with more than three miners}\label{sec:generaltm} 
We fix some $k\geq 4$. We will consider a slight variant of the tree-mining game, where there is initially a single miner in the mine and at all rounds the adversary can add new miners to the root of the mine,  letting the player assign them to any active leaf. The adversary can add at most $k-1$ new miners, so the total number of miners remains bounded by $k$. This modification of the game does not change its analysis because it is always clearly in the adversary's best interest to add all $k$ miners as soon as the game starts. This formulation is helpful to describe a recursive strategy for the player. This variant of the game and the recursive strategy are formally defined in Section \ref{sec:details}. We provide here the key elements of proof.  


As before, the strategy that we denote by $s_p\in \Scal_p$ is composed of epochs. An epoch consists in a sequence of rounds that will have the board go from a $(D,d)$-structure, with integers $d<D$, to some $(D',d')$-structure with integers $D'$ and $d'$ satisfying one of the two following conditions, 
\begin{align}
    D+(D-d)&\leq D',\label{eq:split}\\
    d+(D-d)&\leq d'.\label{eq:join}
\end{align}
The case where the epoch starts with $D=d$ is simple and it is treated separately. 

We now describe the course of one epoch starting from a $(D,d)$-structure.  We assume by induction that a tree-mining strategy capable of handling up to $k-1$  miners was defined. 
\begin{itemize}
    \item[--] \textbf{if} $d < D$ at the start of the epoch \textbf{then} the miners are partitioned into independent groups that correspond to the leaf on which they are located at the start of the epoch. The epoch will end as soon as one of these two conditions is met  \eqref{eq:split} all groups have attained depth $D+(D-d)$ (thus, $D+(D-d)\leq D'$) ; \eqref{eq:join} all robots belong to the same group (thus, $d+(D-d)\leq d'$) ;
    \begin{itemize}
    \item[--] \textbf{for each group}, the player starts an independent recursive instance of the strategy $s_p$, rooted at the leaf where the team is formed -- therefore involving at most $k-1$ miners -- to respond to the adversary's moves. When all miners of some group have attained depth $D+(D-d)$, the corresponding instance becomes \textit{finished}. All the actions that the adversary will address to this group shall have the player respond by reassigning the associated miners to other \textit{unfinished} groups with the smallest number of miners.
    \item[--] \textbf{if} some miner is added to the mine \textbf{then} it is added to the instance of the \textit{unfinished} group with the smallest number of miners.
    \end{itemize}
    \item[--] \textbf{else if} $D=d$ at the start of the epoch \textbf{then} all miners are on the same leaf. The adversary kills this leaf with $c\leq k-1$ children, and the player partitions the miners evenly among them. This leads to a $(D+1,D+1)$-structure, or alternatively to a $(D+1,D)$-structure. The epoch ends after one round.
\end{itemize}
We now analyze the cost of an epoch. While there are at least $2$ \textit{unfinished} groups, the maximum size of a group is $\ceil{k/2}$. During this first phase of the epoch, there are at most $k$ such groups, thus the total cost incurred during this phase is bounded by $kf_{s_p}(\ceil{k/2},D-d)+10k^2(D-d)$, where the extra $10k^2(D-d)$ is an upper bound on the cost of all displacements of the miners between groups 
(the precise decomposition of all such moves is detailed in Section \ref{sec:details}, in particular it is necessary to consider the case where some group starts below depth $D+(D-d)$). During the second phase of the epoch, there is only one \textit{unfinished} subgroup and the corresponding instance uses at most $k-1$ miners. Thus, the cost incurred during this period is bounded by $f_{s_p}(k-1,D'-D)$. This is summarized in the table below. 
\begin{equation*}
\begin{array}{||ccccc||}
\hline
\text{Start} &&\text{Stop} & \text{s.t.} & \text{Max. Cost}\\ 
\hline
     (D,d) & \rightarrow &(D',d') & \eqref{eq:split} \text{ or } \eqref{eq:join}& f_{s_p}(k-1,D'-D)+kf_{s_p}(\ceil{k/2},D-d)+10k^2(D-d)\\
     (D,D) &\rightarrow &(D+1,d') & {d'\in\{D,D+1\}}&  k\\
     \hline
\end{array}
\end{equation*}
Notice that if one of conditions \eqref{eq:split} or \eqref{eq:join} is satisfied, we have $D-d \leq \max\{D'-D,d'-d\}\leq (D'-D)+(d'-d)$. We can thus bound the cost of an epoch that goes from a $(D,d)$-structure to a $(D',d')$-structure by,
\begin{equation}\label{eq:dynamic}
    f_{s_p}(k-1,D'-D)+k\left[f_{s_p}(\ceil{k/2},D'-D)+f_{s_p}(\ceil{k/2},d'-d)\right]+10k^2\left[(D'-D)+(d'-d)\right],
\end{equation}
which leads to the following result.
\begin{theorem}\label{th:mainth}
    The recursive strategy $s_p\in \Scal_p$ defined above satisfies for any $k\geq 2$ and $D\in \Nbb$, 
    \begin{equation}\label{eq:main_eq}
        f_{s_p}(k,D) \leq c_k D,
    \end{equation}
    where $(c_k)_{k\geq 2}$ is defined by
    $
    \begin{cases}
        c_2=2,\\
        c_k = c_{k-1}+2kc_{\ceil{k/2}}+20k^2,
    \end{cases}$ and satisfies
    $c_k = \mathcal{O}(k^{\log_2 k}).$
\end{theorem}
We prove equation \eqref{eq:main_eq} for all values of $k\geq 2$ by induction. The initialization with $k=2$ is a direct consequence of Proposition~\ref{prop: kequals2} above. We fix $k\geq 3$ and assume the equation holds for all values of $k'\in \{2,\dots, k-1\}$.
\begin{lemma}\label{lemma:lemma:generalk} If an epoch ends with a $(D,d)$-structure, the cost is at most $a_k D+b_k d$, with $a_k = c_{k-1}+kc_{\lceil k/2\rceil}+10k^2$ and $b_k = kc_{\lceil k/2\rceil}+10k^2$.
    \end{lemma}
\begin{proof}    
    We consider some $(D',d')$-structure attained at the end of an epoch, and we reason by induction over epochs. The epoch must have started from a $(D,d)$-structure, for which that the cost was below 
    $a_kD+b_kd$ by induction.
    By equation~\eqref{eq:dynamic} and equation \eqref{eq:main_eq}; we have that the cost of an epoch going from a $(D,d)$-structure to a $(D',d')$-structure is at most
    $a_k(D'-D)+b_k(d'-d),$
    thus completing the induction. 
    \end{proof}

\begin{lemma}
    Under the strategy $s_p$, for any depth $D\in \Nbb$, we have that $f_{s_p}(k,D) \leq c_k D$, where $c_k = a_k+b_k = c_{k-1}+2kc_{\ceil{k/2}}+20k^2$.
\end{lemma}

\begin{proof}
    We consider the cost at a point of the game where some miner is at depth $D'\leq D$. We assume that the adversary plays by killing all leaves, except for that specific leaf. By Lemma \ref{lemma:lemma:generalk}, the epoch therefore ends with a $(D',D')$-structure corresponding to a cost below $(a_k+b_k)D.$
\end{proof}
This finishes to prove the recurrence relation for $(c_k)_{k\geq 2}$. The analysis of its asymptotic behaviour is given by the following lemma, which is proved in Section \ref{sec: calculus}. 
\begin{lemma}[see Lemma \ref{lemma:asymptotics}]
    The sequence $(c_k)_{k\geq 2}$ defined above satisfies
    $c_k = \mathcal{O}(k^{\log_2 k}).$
\end{lemma}

\section{From Collective Tree Exploration to Tree-Mining}\label{sec: red}
In this section, we establish the connection between collective tree exploration (\cte) and tree-mining game (\tm). For this purpose, we will introduce an intermediary setting that we call \textit{asynchronous collective exploration} (\acte) generalizing collective tree exploration and extending the adversarial setting of \cite{cosson2023efficient}. In Section \ref{sec:red1}, we present a reduction from (\cte) to (\acte) and in Section~\ref{sec:newred} we present the reduction from (\acte) to (\tm). A reciprocal reduction from (\tm) to (\acte) is finally presented in Section \ref{sec:reciprocal}. 
This section strongly relies on the notion of locally-greedy exploration algorithms introduced in Section \ref{sec:locallyg}. The main result of this section is the following.
\begin{theorem}\label{th:mainreduction} Any $f(k,D)$-bounded strategy for the tree-mining game induces induces a collective tree exploration algorithm $\Acal$ satisfying on any tree $T$ with $n$ nodes and depth $D$,  
$${\normalfont\texttt{Runtime}}(\Acal,k,T) \leq \left\lceil\frac{2n+f(k,D)}{k}\right\rceil +D.$$
\end{theorem}

\subsection{Asynchronous Collective Tree Exploration}\label{sec:red1}
In this section, we define the problem of asynchronous collective tree exploration (ACTE) which is a generalization of collective tree exploration (CTE) in which agents move sequentially. 
This setting provides a strong motivation for the aforementioned \textit{locally-greedy algorithms}. It also allows to further relate collective tree exploration to the classic framework of \textit{online} problems. 

Specifically, we assume that at each step $t$, only one robot indexed by $r_t\in [k]$ is allowed to move. We make no assumption on the sequence $r_0,r_1,\dots$ that can be chosen arbitrarily by the environment. Though robots have unlimited communication and computation capabilities, we slightly reduce the information to which they have access. 
At the beginning of move $t$, the team/swarm is only given the additional information of whether robot $r_t$ is adjacent to an unexplored edge $e$, and if so, the ability to move along that edge. Note that in contrast with collective tree exploration, the team may not know the exact number of edges that are adjacent to some previously visited vertices. We will say that a node has been \textit{mined} if all edges adjacent to this node have been explored and the information that this node is not adjacent to any other unexplored edge has been revealed to the team. The terminology is evocative of the tree-mining game for reasons that will later become apparent. Note that a node that is not mined is possibly adjacent to an unexplored edge.

Exploration starts with all robots located at the root of some unknown tree $T$ and ends only when all nodes have been discovered and mined. We do not ask that all robots return to the root at the end of exploration for this might not be possible in finite time (if say some robot breaks-down indefinitely while it is away from the root). For an asynchronous exploration algorithm $\Bcal$, we denote by $\texttt{Moves}(\Bcal,k,T)$ the maximum number of moves that are needed to traverse all edges $T$ for \textit{any} sequence $r_0,r_1, \dots$ of allowed robot moves. Asynchronous collective tree exploration (ACTE) generalises collective tree exploration (\cte) as follows,
\begin{theorem}
    For any asynchronous exploration collective tree exploration algorithm $\Bcal$, one can derive a collective tree exploration algorithm $\Acal$ satisfying,
    \begin{equation*}
    {\normalfont\texttt{Runtime}}(\Acal,k,T) \leq \left\lceil\frac{1}{k}{\normalfont\texttt{Moves}(\Bcal,k,T)}\right\rceil+D,
    \end{equation*}
    for any tree $T$ of depth $D$. 
\end{theorem}

\begin{proof}
    We assume that we are given access to some (ACTE) algorithm $\Bcal$ and we query this algorithm with the infinite sequence of robot moves $1,2, \dots, k,1,2, \dots, k, 1,2, \dots$. We now define the moves of all robots for a (\cte) algorithm $\Acal$ at a synchronous round $\tau\in \Nbb$ as follows: if the tree is not fully explored, each robot $r\in [k]$ performs the move $\tau k+r$ of algorithm $\Bcal$ ; and if the tree is fully explored then all robots head to the root. We now make the two following statements,
    \begin{enumerate}
        \item Algorithm $\Acal$ is well defined, i.e. at round $\tau$ of algorithm $\Acal$ the robots have access to enough information to emulate the running of $\Bcal$ from move $k\tau+1$ to move $k(\tau+1)$.
        \item Algorithm $\Acal$ runs in $\left\lceil\frac{1}{k}\texttt{Moves}(\Bcal,k,T)\right\rceil+D$ time steps at most.
    \end{enumerate}
    The first statement is obvious from the fact that at round $\tau$ of algorithm $\Acal$, the robots have access to the port numbers of all the edges that are currently adjacent to one of the agents and that information is sufficient because each agent moves only once from round $\tau$ to round $\tau+1$. The other statement is directly implied by the fact that after $\left\lceil\frac{1}{k}\texttt{Moves}(\Bcal,k,T)\right\rceil$ rounds all edges of the tree have been visited and the robots will thus reach the root in only $D$ additional synchronous rounds. 
\end{proof}

\subsection{Locally-greedy exploration algorithms}\label{sec:locallyg} In this section, we study the notion of locally-greedy algorithm. A locally-greedy algorithm is one in which a robot will always prefer to traverse an unexplored edge rather than a previously explored edge. Formally, we define a locally-greedy algorithm as follows,
\begin{Definition}[Locally-Greedy ACTE]
A locally-greedy algorithm for asynchronous collective tree exploration is one such that at any move $t$, if robot $r_t$ is adjacent to an unexplored edge, then it selects an unexplored edge as its next move.
\end{Definition}
Note the notion generalises to the setting of synchronous collective tree exploration (CTE). 
\begin{Definition}[Locally-Greedy CTE] A locally-greedy collective tree exploration algorithm is such that at any synchronous round, if a node is adjacent to $e$ unexplored edges and is populated by $x$ robots, then at the next round, this node will be adjacent to exactly $\max\{0,e-x\}$ unexplored edges. 
\end{Definition}
We now define the notion of \textit{anchor} of a robot in a locally-greedy algorithm. 
\begin{Definition}[Anchors of a Locally-Greedy Algorithm] For any locally-greedy algorithm, for each robot $r\in [k]$, we define the anchor $a_t(r) \in V$ as the highest node that is not mined in the path leading from robot $r$ to the root. If there is no such node, the anchor is considered empty $\perp$. 
\end{Definition}
\begin{proposition}
    In a locally-greedy algorithm, at any stage of the exploration, all undiscovered edges are below an anchor. Consequently, exploration is completed when all anchors are equal to $\perp$.
\end{proposition}
\begin{proof}
    Consider some undiscovered edge. Consider the lowest discovered ancestor of this edge. That node is not mined. Consider the robot that first discovered that node. That robot may not have left the sub-tree rooted at this node without having explored the edge that leads to the undiscovered edge.
\end{proof}
The definition of locally-greedy algorithms does not specify the move of the selected robot $r_t$ if it is not adjacent to an unexplored edge. We now define an extension of locally-greedy exploration algorithms in which all robots maintain a target towards which they perform a move when they are not adjacent to unexplored edges.

\begin{Definition}[Locally-Greedy Algorithm with Targets]
    A locally-greedy algorithm with targets maintains at all times a list of explored nodes $(v_t(r))_{r\in [k]}$ called targets such that the $t$-th move is determined by the three following rules,
\begin{enumerate}[label=R\arabic*.]
    \item \textbf{if} robot $r_t$ is adjacent to an unexplored edge \textbf{then} it selects that edge as its next move, \label{it: 1}
    \item \textbf{else} it selects as next move the adjacent edge that leads towards its target $v_t(r_t)$. \label{it: 2}
    \item Robot $r_t$ does not station at its location. \label{it: 3}
\end{enumerate}
\end{Definition}
We say that robot $r$ is targeting at $v_t(r)$ but we note that this does not mean that the robot is necessarily located below its target, as targets are different from anchors. 
An equivalent perspective on rule \eqref{it: 3} is that the value of the target of robot $r_t$ must be modified at the beginning of move $t$ if the following condition holds. 
\begin{enumerate}[label=C.]
    \item Robot $r_t$ is located at its target $v_{t-1}(r_t)$ and is not adjacent to an unexplored edge. \label{it: change_target}
\end{enumerate}
Note that any locally-greedy algorithm can be viewed as having a targets. We will thus particularly focus on algorithms for which the total movement of the targets, $\sum_{r\in [k]}\sum_{t< M} d(v_t(r),v_{t+1}(r))$ is small. This focus is motivated by the following proposition.

\begin{proposition}\label{prop:targets} After $M$ moves, a locally-greedy algorithm with targets has explored at least,
    $$ \frac{1}{2}\left(M - \sum_{i\in [k]}\sum_{t< M} d(v_t(r),v_{t+1}(r))\right)$$
edges, where $v_t(r)$ denotes the target of robot $r$ at move $t$ and $d(\cdot, \cdot)$ is the distance in the underlying tree. 
\end{proposition}
To prove the result, we present a lower-bound on the number of edges explored by a single agent involved in a locally-greedy algorithm with target. Applying the lemma below to each individual robot then yields Proposition~\ref{prop:targets}. 
    
\begin{lemma}
    Any robot starting at the root and performing $m$ moves at some instants within $\{1,\dots, M\}$ prescribed by rules \eqref{it: 1} or \eqref{it: 2}
    uses at least 
    \begin{equation*}
        \frac{1}{2}\left(m - \sum_{t<M} d(v_t,v_{t+1})\right)
    \end{equation*}
    applications of rule \eqref{it: 1}, where $(v_t)_{t\in [M]}$ denotes the list of consecutive targets used by the robot, and $v_1 = {\normalfont\texttt{root}}$. 
\end{lemma}

\begin{proof}
    Without loss of generality, we focus only on the moments at which the robot is allowed to move. We re-index all such instants by $t\in [m]$. We then consider the quantity,
    \begin{equation*}
        d_t = d(p_t, v_t),
    \end{equation*}
    where $p_t$ denotes the position of the robot right before move $t$. The quantity is initialised with $d_0 =0$ because at initialisation $p_0=v_0=\texttt{root}$. We note that this $d_t$ is always non-negative. We then write the update rule,
    \begin{align*}
        d_{t+1}-d_t 
        &= d(p_{t+1},v_{t+1})-d(p_t,v_t)\\
        &=d(p_{t+1},v_{t+1})-d(p_{t+1},v_{t})+d(p_{t+1},v_t)-d(p_t,v_t)\\
        &=d(p_{t+1},v_{t+1})-d(p_{t+1},v_{t})-\1(t: \ref{it: 2})+\1(t: \ref{it: 1})
\end{align*}
Where we denote by $\1(t: \ref{it: 1})$ (resp $\1(t: \ref{it: 2})$) the indicator of whether rule $\ref{it: 1}$ (resp \eqref{it: 2}) is used at move~$t$. We used the fact that every call to rule \eqref{it: 2} reduces the distance of the robot to its target by exactly one, whereas all calls to rule \eqref{it: 1} result in an increase that same distance by exactly one. Using now that $\forall t: \1(t: \ref{it: 2})+\1(t: \ref{it: 1})=1$, we get
$$d_{t+1}-d_t = d(p_{t+1},v_{t+1})-d(p_{t+1},v_{t})-1+2\1(t: \ref{it: 1})$$
which yields by telescoping,
$$d_m = \sum_{t<m} \left(d(p_{t+1},v_{t+1})-d(p_{t+1},v_{t})\right) - m + 2\#\{t : \ref{it: 1}\},$$
and thus, 
$$\#\{t : \ref{it: 1}\} = \frac{1}{2}\left(m +d_m - \sum_{t<m} \left(d(p_{t+1},v_{t+1})-d(p_{t+1},v_{t})\right)\right).$$
The simple observation that $d_m\geq 0$, and the triangle inequality $d(p_{t+1},v_{t+1})- d(p_{t+1},v_t)\leq d(v_t,v_{t+1})$ then gives the desired result. 
\end{proof}
It is clear from the proof above that the slightly more general statement below holds. 
\begin{proposition} After $M$ moves, a locally-greedy algorithm with targets has explored at least, 
    $$ \frac{1}{2}\left(M - \sum_{r\in [k]}\sum_{t< M} d(p_{t+1}(r),v_{t+1}(r))-d(p_{t+1}(r),v_{t}(r))\right)$$
where $v_t(r)$ denotes the target of robot $r$ at move $t$, and $p_t(r)$ its position, and $d(\cdot, \cdot)$ is the distance in the tree. 
\end{proposition}

\subsection{Asynchronous Collective Tree Exploration via Tree-Mining}\label{sec:newred}
Locally-greedy algorithms are natural candidates to perform asynchronous collective tree exploration. In this section, we shall see that a tree-mining strategy entirely defines a locally-greedy exploration algorithm, and provides a guarantee on its performance.

\begin{theorem}\label{th:redtm} A tree-mining strategy $s_p \in \Scal_p$ induces a locally-greedy asynchronous collective tree exploration algorithm $\Bcal$ satisfying,  
$$ {\normalfont\texttt{Moves}(\Bcal,k,T)} \leq 2n +f_{s_p}(k,D),$$
for any tree $T$ with $n$ nodes and depth $D$.
\end{theorem}

\begin{proof} We assume that we are given a tree-mining strategy $s_p$ and we design an asynchronous collective exploration algorithm $\Bcal$. The algorithm $\Bcal$ is a locally-greedy algorithm with targets, and the update of targets is done using the tree-mining strategy. The pseudo-code of $\Bcal$ is formally given below in Algorithm \ref{alg:tream}. More specifically, at all moves~$t>0$, the algorithm maintains a target $v_t(r)$ for each robot $r\in [k]$. At the start of the algorithm, all robots are assumed to be mining at the root, i.e., $\forall r\in [k] : v_0(r) = \texttt{root}$. At move~$t$, before robot $r_t$ is assigned a move, the list of targets is updated if and only if the following condition is met,
\begin{enumerate}[label=C.]
    \item Robot $r_t$ is located at its target $v_{t-1}(r_t)$ and is not adjacent to an unexplored edge. 
\end{enumerate}
One may notice that when condition \eqref{it: change_target} is satisfied, neither of rules \eqref{it: 1} or \eqref{it: 2} can be applied straightforwardly, since the robot that shall move is already located at its target. 

We now describe the update of the targets $(v_{t-1}(r))_{r\in [k]}$ at move $t$. Let $\tau$ the number of times condition \eqref{it: change_target} was raised so far in the exploration. Let $L(\tau)$ be the set of all current targets and $T(\tau)$ be the tree given by all current and previous targets (the fact that they form a tree containing the root will be obvious from what follows). For any target $v\in L(\tau)$ we let $x_v(\tau)$ be the number of robots targeting $v$, and observe that $\bx(\tau) = (x_v(\tau))_{v \in L(\tau)}$ is a configuration. Also, for shorthand, we denote by $u := v_{t-1}(r_t)$ the present target of robot $r_t$, which triggered condition \eqref{it: change_target}. Finally, we decompose the $x_u(\tau)$ robots targeted at $u$ into two groups,
\begin{itemize}
    \item[--] $c$ robots that are located on descendants on $u$, but not on $u$; 
    \item[--] $x_u(\tau)-c$ robots that are not located on descendants of $u$,
\end{itemize}
We note that $\Tcal(\tau) = (T(\tau),\bx(\tau))$ forms a board in the sense of the tree-mining game. We also observe that $(u,c)$ is a valid choice for the adversary, since $u$ is an active leaf and $c\leq x_u(\tau)-1$ because at least one robot with target $u$ is located at $u$. Using the strategy $s_p$ we obtain a new configuration $\bx(\tau+1)$ as the response of the player. 
The targets $(v_{t-1}(r))_{r\in [k]}$ are then updated as follows. All $c$ robots that were exploring below $u$ are assigned as target the children of $u$ that is the parent to the branch that they are already exploring. The remaining $x_u(\tau)-c$ robots are re-targeted arbitrarily such as to respect the new load $\bx(\tau+1)$ prescribed by the player. 

\subsubsection*{Analysis of the number of moves} We now turn to the analysis of the maximum number of moves required by the algorithm described above to complete exploration. By property of target-based algorithms, after $M$ exploration moves, the number of edges that have been discovered is at least, $\frac{1}{2}(M - S)$, where $S$ is incremented by 
$\sum_{r\in [k]} d(p_{t+1}(r),v_{t+1}(r))-d(p_{t+1}(r),v_{t}(r)) = d(\bx(\tau),\bx(\tau+1))-2c$, when the targets are updated.
Note that this quantity is the increment of the cost in the associated tree-mining game. Also observe that the corresponding tree-mining game has at least one active leaf (in fact, all) below depth $D$. Thus we have $S \leq f_{s_p}(k,D)$. Since the total number of edges is bounded by $n$, we have $M-S \leq 2n$. The total number of moves $M$ before exploration is completed must thus satisfy
$M \leq 2n +f_{s_p}(k,D),$
which concludes the proof. 
\end{proof}

\begin{algorithm}[H]
\caption{\texttt{TEAM} (Tree-Mining Exploration Algorithm)}\label{alg:tream}
\begin{algorithmic}[1]
\Require A strategy for the tree-mining game.
\Ensure Prescribes robot moves until exploration.
\State $\forall r\in \{1,\dots,k\}: v(r) \gets \texttt{root}$ \Comment{Robots are initially located and targeted at the root.}
\While{$t\geq 0$}
\If{robot $r_t$ is adjacent to unexplored edge $e$} 
\State $\texttt{MOVE-ALONG}(r_t,e)$ \Comment{apply rule \eqref{it: 1}}
\Else
\State $u\gets v(r_t)$
\If{robot $r_t$ is located at $u$}
\State $c \gets $ the number of robots that are strictly below $u$
\State $\bx \gets $ configuration representing number of robots assigned to each target
\State $T \gets$ board tree of all past targets, where current targets are the active leaves 
\State $\bx' \gets $ the response of the strategy to move $(u,c)$ in board $\Tcal = (T,\bx)$
\State \texttt{Modify} the targets $(v(r))_{r\in [k]}$ to follow configuration $\bx'$ 
\EndIf
\State $\texttt{MOVE-TOWARDS}(r_t,v(r_t))$ \Comment{apply rule \eqref{it: 2}}
\EndIf
\EndWhile
\end{algorithmic}
\end{algorithm}

\subsection{Tree-Mining via Asynchronous Collective Tree Exploration}\label{sec:reciprocal}
In this section, we show a reciprocal to Theorem \ref{th:redtm}. It is readily applied in Proposition \ref{prop:lower-bound} to obtain a lower-bound on the overhead of asynchronous collective tree exploration. 
\begin{theorem}
For any asynchronous locally-greedy collective exploration algorithm $\Bcal$ with additive overhead below $f(k,D)$, i.e. satisfying ${\normalfont\texttt{Moves}(\Bcal,k,T)} \leq 2n +f(k,D)$ for any tree $T$ with $n$ nodes and depth $D$, 
one can define a $f(k,D)$-bounded tree-mining strategy $s_p\in \Scal_p$.
\end{theorem}
\begin{proof}
Consider an asynchronous collective exploration algorithm $\Bcal$. We shall use this algorithm to explore a graph that corresponds to the board of a tree-mining game. Assume that at some round $\tau$ of the game, the board $\Tcal(\tau)$ has $L(\tau)$ active leaves on which $k$ robots are located, following a configuration $\bx(\tau)$. At this round, we assume that the adversary selects some node $\ell\in L(\tau)$ and provides it with $c$ children. We will design the response of the player using the exploration algorithm $\Bcal$. 
We assume that the adversary grants a move to all robots in $u$, providing an unexplored edge to each of the first $c$ robots, and no unexplored edge to the subsequent robots. The robots that were located at $u$ are given a move until they reach an active leaf in $L(\tau+1)$ and then they are blocked. It must be the case that the moves will stop because the additive overhead of $\Bcal$ is finite. When all robots have all reached $L(\tau+1)$, we note that the final configuration forms a response for the player $\bx(\tau+1)$. Also, the movement cost payed by the exploration algorithm $\Bcal$ is at least that paid by the tree-mining game. Consequently, we have $f_{s_p}(k,D)\leq f(k,D)$.
\end{proof}

\begin{proposition}\label{prop:lower-bound}
    For any strategy $s_p\in \Scal_p$ of the player, and any $k,D \in \Nbb: k(\ln(k)-4)D \leq f_{s_p}(k,D).$
    Consequently, this is a lower-bound on the overhead of any asynchronous locally-greedy exploration algorithm. 
\end{proposition}
\begin{proof}
We present a strategy of the adversary $s_a\in \Scal_a$ satisfying $k(\ln(k)-4)D \leq \min_{s_p\in \Scal_p}\cost(s_p,s_a,k,D)\leq \min_{s_p\in \Scal_p} f_{s_p}(k,D)$. First, the adversary chooses the root and provides it with $k-1$ children. Then it successively kills each of these children, always choosing the one which has the highest number of miners, until there is only one remaining active leaf. Then the adversary repeats this phase until it reaches depth $D$. Note that when there are $2\leq h \leq k-1$ active leaves, the leaf chosen by the adversary contains at least $\ceil{k/h}$ miners and yields a cost of $2$ each. Thus the increase of the cost during one phase is at least $\sum_{h=2}^{k-1}2\ceil{k/h}-2(k-1)\geq k\int_3^k\frac{1}{h}dh-2(k-1)\geq k(\ln(k)-\ln(3)-2)\geq k(\ln(k)-4)$. 
\end{proof}

\subsection*{Acknowledgements} The author thanks Laurent Massoulié and Laurent Viennot for stimulating discussions. 
\bibliography{biblio}
\appendix
\section{Calculations used in Section \ref{sec: game}}\label{sec: calculus}
\begin{lemma}\label{lemma: recursion}
    For some fixed reals $\alpha,\beta,\gamma, c_1,c_2>0$ satisfying $\alpha\geq \frac{\beta+1}{2\ln(2)}$, there exists a constant $k_0:=k_0(\alpha,\beta,\gamma,c_1,c_2)$ such that the quantity, 
    $$u_k = \exp(\alpha \ln k^2) = k^{\alpha \ln k},$$
    satisfies the relation, 
    $$\forall k\geq k_0: u_{k-1}+c_1k^\beta u_{\ceil{k/2}}+c_2k^\gamma\leq u_k.$$
\end{lemma}
\begin{proof}
For clarity, we chose to prove the bound $u_{k}+c_1(k+1)^\beta u_{\ceil{(k+1)/{2}}}+c_2(k+1)^\gamma\leq u_{k+1}$. The computation of $u_{k+1}$ goes as follows,
\begin{align*}
u_{k+1} &= \exp(\alpha \ln(k+1)^2)\\
&= \exp(\alpha (\ln k+\ln(1+1/k))^2)\\
&\geq \exp(\alpha \ln k^2+2\alpha\ln k\ln(1+1/k))\\
&\geq \exp(\alpha \ln k^2)(1+2\alpha\ln k\ln(1+1/k))\\
&\geq u_k\left(1+\frac{\alpha\ln k}{2k}\right), \quad \text{for } k\geq 3.
\end{align*}
We also have, since $\ceil{(k+1)/2}\leq k/2+1$,
\begin{align*}
    u_{\ceil{(k+1)/2}} 
    &\leq \exp(\alpha \ln(k/2+1)^2)\\
    &\leq \exp(\alpha (\ln k+\ln(1+2/k)-\ln(2))^2)\\
    &\leq \exp(\alpha (\ln k-\ln(2))^2+2\alpha \ln k\ln(1+2/k)) \\
    &\leq \exp(\alpha (\ln k-\ln(2))^2)\exp(4\alpha \ln k/k)\\
    &\leq \exp(\alpha \ln k^2-2\alpha \ln(2)\ln k+\alpha\ln(2)^2) \exp(4\alpha)\\
    &\leq \exp(5\alpha) u_k/k^{2\alpha \ln(2)}.
\end{align*}
When putting these equations together, we obtain,
\begin{align*}
     u_k + (k+1)^\beta u_{\ceil{(k+1)/2}} +(k+1)^2 \leq u_{k}\left(1 + \mathcal{O}\left(\frac{1}{k^{2\alpha\ln(2)-\beta}}\right)+\frac{(k+1)^2}{u_{k}}\right).
\end{align*}
Since we clearly have $\frac{(k+1)^2}{u_{k}} = \mathcal{O}(\frac{1}{k})$ and since by assumption $2\alpha\ln(2)-\beta \geq 1$, we obtain,   \begin{align*}
     u_k + (k+1)^\beta u_{\ceil{(k+1)/2}} +(k+1)^2 \leq u_{k}\left(1 + \mathcal{O}\left(\frac{1}{k}\right)\right),
\end{align*}
which is sufficient to conclude.
\end{proof}

\begin{lemma}\label{lemma:asymptotics}
    The sequence $(c_k)_{k\geq 2}$ defined by,
    $
    \begin{cases}
        c_2=2,\\
        c_k = c_{k-1}+2kc_{\ceil{k/2}}+20k^2,
    \end{cases}$ satisfies
    $c_k = \mathcal{O}(k^{\log_2 k}).$
\end{lemma}
\begin{proof}We can use the lemma above, taking $\alpha = \frac{1}{\ln(2)}, \beta= 1, \gamma=2$ which define a constant $k_0(\alpha,\beta,\gamma)$. We then take $C = \max\{c_1,\dots, c_{k_0}\}$. The property is initialized for all values of $3\leq k\leq k_0$, it is then clear by induction that $u_k\leq Ck^{\log_2 k}$, for all values $k\geq 2$.    
\end{proof} 
\section{Detailed construction of the tree-mining strategy}\label{sec:details}

In this section we slightly complexity the rules of the tree-mining game, to allow for a rigorous definition of the tree-mining strategy described in Section \ref{sec:generaltm}. These modifications lead to an equivalent variant of the game called the \textit{extended tree-mining game}. 

\subsection{Extended Tree-Mining Game}
The extended tree-mining game enjoys the same definitions as the regular tree-mining game described in Section \ref{sec: game-description}, along with two additional rules,
\begin{enumerate}
[label=E\arabic*.]
    \item The adversary may add a miner to the mine. \hfill (New Miners)
    \item At all rounds, the player may move miners between active leaves. \hfill (Non-Lazy Moves)
\end{enumerate}

\subsubsection*{New miners} As was briefly mentioned in Section \ref{sec: game}, it will be useful to assume that the game starts with a single miner located at the root, and that the adversary may add (but not withdraw) a miner to the mine at any point in time. When the adversary adds a miner to the mine, the player may insert it at any currently active leaf without any movement cost. The adversary may choose to start by adding $k$ miners and never add another miner later in the game; in this case, the strategy of the player can be seen as a strategy for the setting where the number of miners is fixed. 

\subsubsection*{Non-lazy moves} It is sometimes convenient for the player to move a miner from an active leaf to another, although neither of these leaves have been selected by the adversary. These moves are called ``non-lazy''. Non-lazy moves induce a movement cost equal to the total distance travelled by the miners. Non-lazy moves have a simple interpretation in terms of their corresponding exploration algorithm as defined in Section \ref{sec: red}. A non-lazy move from a leaf/target $v$ to a leaf/target $w$ can be implemented by re-targetting immediately at $w$ the first robot $r$ to mine $v$, with no trigger of condition \eqref{it: change_target}. Non-lazy moves can thus be seen as a payment in advance of a future movement of a robot to a new target. For practicality, we restrict non-lazy moves in that there must always be at least one miner in all active leaf.

We now define $\Scal_p(k)$, $\Scal_a(k)$ the set of strategies of the player and the adversary in the extended tree-mining game, when at most $k$ miners are involved in the game. For a given strategy of the adversary $s_a \in \Scal_a(k)$, we recall that $\cost(s_p,s_a,k,D)$ is defined as the maximum cost attained while some miner is at depth less or equal to $D$, if the adversary and the player use the pair of strategies $(s_p,s_a)$. Given a strategy $s_p\in \Scal_p(k)$, we define as follows, 
\begin{align*}
f_{s_p}(k,D) &=\max_{s_a \in \Scal_a(k)} \cost(s_p,s_a,k,D).
\end{align*}
It is possible to show that the extended tree-mining game is not easier for the player than the regular tree-mining game. In particular, we have the following extension of Theorem~\ref{th:redtm}.

\begin{theorem} A strategy $s_p \in \Scal_p(k)$ for the extended tree-mining game induces a locally-greedy asynchronous collective tree exploration algorithm $\Bcal$ satisfying, 
$$ {\normalfont\texttt{Moves}(\Bcal,k,T) \leq 2n +f_{s_p}(k,D)},$$
for any tree $T$ with $n$ nodes and depth $D$.
\end{theorem}
\begin{proof}
    We assume that we are given a strategy $s_p \in \Scal_p(k)$ for the extended tree-mining game. The exploration algorithm that we consider is entirely identical to that described in Section \ref{sec:newred}. The board starts with $k$ miners at the root, and the strategy $s_p$ is called every time condition \eqref{it: change_target} is satisfied.  If the strategy $s_p$ performs a non-lazy move at some instant $t$, we re-target at $w$ the next robot that mines $v$, before it triggers condition \eqref{it: change_target}. The rest of the proof is unchanged.
\end{proof}

\subsection{Bounded-Horizon Extended Tree-Mining Game}

The strategies that we have described so far are designed to aim for an infinite horizon, represented by the variable $D$ that can take any values in $\Nbb$. In our construction it will be useful to consider strategies with a bounded horizon $\Delta\in \Nbb$. For this purpose, we define here the $\Delta$-horizon tree-mining game, which is a variant of the extended tree-mining game with the following additional rules,
\begin{enumerate}[label=F\arabic*.]
    \item The player may only place $1$ miner on leaves at depth $\Delta$,
    \item The game ends when all active leaves are at depth $\Delta$. 
\end{enumerate}
The following proposition shows that a strategy for the extended tree-mining game induces a strategy for the variant with bounded horizon $\Delta$. 
\begin{proposition}
    For any strategy $s_p\in \Scal_p(k)$ of the extended tree-mining game, there exists a strategy $s_{p,\Delta}\in \Scal_p(k,\Delta)$ of the $\Delta$-horizon tree-mining game, such that the cost at the end of the game is less than $f_{s_p}(k,\Delta)$. Moreover, for any depth $D\leq \Delta$, we have, $ f_{s_{p,\Delta}}(k, D)\leq f_{s_p}(k, D)$ where $f(\cdot)$ is defined for the $\Delta$-horizon game analogously as for the regular game. 
\end{proposition}
\begin{proof}
We start by defining $s_{p,\Delta}\in\Scal_p(k,\Delta)$ using the unbounded horizon strategy $s_p\in \Scal_p(k)$. Throughout, we will consider a board $\Tcal$ of the unbounded horizon tree-mining game and a board $\Tcal_\Delta$ of the bounded-horizon tree-mining game. Both board are reduced to one node at the start. For every move of the adversary on $\Tcal_\Delta$, we will emulate one or multiple moves of the adversary on $\Tcal$ and use the response of $s_p$. We will maintain that $\Tcal_\Delta$ remains isomorphic to $\Tcal$ when restricted to nodes at depth less than or equal to $\Delta$. We now describe the response of the player on $\Tcal_\Delta$ after the move of the adversary. We first mimic the move of the adversary on the tree $\Tcal$ and consider the two following cases: 
\begin{itemize}
    \item[--] The player $s_p\in \Scal_p(k)$ responds on $\Tcal$ by sending miners to nodes at depth less than or equal to $\Delta$. After this move, $\Tcal$ satisfies the fact that no node at depth more than $\Delta$ contains more than one miner. In this case, the player $s_{p,\Delta}$ can respond exactly like $s_p$.
    \item[--] The player of the unbounded game $s_p\in \Scal_p(k)$ responds in a way such that some node of $\Tcal$ at depth more than $\Delta$ contains contains more than two miners. We then emulate that the adversary will query this node and provide it with exactly one child in $\Tcal$. We iterate this extension until all nodes below $\Delta$ in $\Tcal$ have $1$ miner exactly. This must happen in finite time since $f_{s_p}(k,\Delta)<\infty$. When this is the case, $s_{p,\Delta}$ is defined as the move in $\Tcal_\Delta$ that allows to maintain the isomorphism with $\Tcal$ for all nodes at depth less than or equal to $\Delta$.
\end{itemize}
For the above strategy $s_{p,\Delta}\in \Scal_p(k,\Delta)$, it is obvious that at any instant the cost of $\Tcal_\Delta$ is less than the cost of $\Tcal$ because all moves in $\Tcal_\Delta$ also took place in $\Tcal$, thus for any $D\leq \Delta, f_{s_{p,\Delta}}(k,D)\leq f_{s_p}(k, D).$
\end{proof}

\subsection{Construction of the Recursive Strategy}
We can now fully construct by induction the strategies $\{s^{(k)}\in \Scal_p(k)\}$ of the extended tree-mining game. The induction is initialised with the simple strategy $s^{(2)}$ described in Section \ref{sec:ksmall}. 
We now assume by induction that the strategies $\{s^{(k')} : k'<k\}$ have been defined and turn to the definition of $s^{(k)}$. The strategy works by iterating epochs as described in Section \ref{sec:generaltm}. An epoch consists in a sequence of moves going from a $(D,d)$-structure, with integers $d<D$, to some $(D',d')$-structure with integers $D'$ and $d'$ satisfying one of the two following conditions, 
\begin{align}
    D+(D-d)&\leq D',\label{eq:split2}\\
    d+(D-d)&\leq d'.\label{eq:join2}
\end{align} 
For simplicity, we will now denote $D-d$ by $\Delta$. We describe an epoch of $s^{(k)}$ that starts from some arbitrary $(D,d)$-board:
\begin{itemize}
\item[--] First, if the board contains leaves $\ell$ at depth $D_{\ell} \geq D+ \Delta$, those leaves will have their population reduced to $1$ miner through non-lazy moves. The corresponding cost of at most $k(D_{\ell}-d+\Delta)$ will not be taken into account, because as we shall see below, we can assume that it was provisioned by the previous epoch.  
\item[--] Next, the load of all the active leaves $\ell$ with depth $D_\ell \leq D+\Delta-1$ is balanced by non-lazy moves such that the value of $k_t(\cdot)$ differs by one at most on all such leaves. Thus, if there at least two such leaves, their load is at most $\ceil{k/2}$. This step induces a movement cost of at most $4k\Delta$ since any two such leaves may be distant by $4\Delta$ and at most $k$ robots have to move. 
\item[--] Then, for all the active leaves $\ell$ with depth $D_\ell \leq D+\Delta-1$, an instance of $s^{(k-1)}$ denoted by $s^{(k-1)}_\ell$ is started at $\ell$ with bounded horizon $\Delta_\ell = \Delta-(D_\ell-D)$. The current epoch will end whenever one of the two following conditions is met \eqref{eq:split2} all instances are \textit{finished}, i.e. they have reached their horizon ; \eqref{eq:join2} all miners are in the same instance. While the epoch is not finished, we apply the following rules, 
\begin{itemize}
\item When an instance $s^{(k-1)}_\ell$ attains its horizon, i.e. when it finishes, all excess robots are evenly distributed in other \textit{unfinished} instances, making sure that the number of robots differs by at most one across unfinished instances. There is a total of at most $k^2$ such reassignments, thus, the corresponding movement cost is bounded by $4k^2\Delta$. 
\item When a new miner arrives, it is added to the \textit{unfinished} instance $s^{(k-1)}_{\ell}$ with the smallest load, again, the number of miners across unfinished instances differs by at most one. Remember that there is no movement cost associated to the arrival of a new miner. 
\end{itemize}
\item[--] Finally, if the epoch stops with condition \eqref{eq:join2}, it may be the case that the final $(D',d')$-structure contains a leaf $\ell'$ at depth $D_{\ell'}\geq D'+\Delta'$ where $\Delta' = D'-d'$. Note that however, $D_{\ell'}\leq D+\Delta$ since we used bounded horizon strategies. Thus the current epoch can provision a budget of $2k\Delta\geq k(D'_{\ell}-d)$ to cover the retrieval of these miners in future epochs. 
\end{itemize}

\begin{proposition}
    The infinite-horizon strategy $s^{(k)}\in \Scal_p(k)$ defined above satisfies that for any $k\geq 2$ for any $D\in \Nbb$, and for any $k'\leq k$,
    $f_{s^{(k)}}(k',D) \leq c_{k'} D,$
    where $(c_k)_{k\geq 2}$ is defined by
    $c_2=2$ and $c_k = c_{k-1}+2kc_{\ceil{k/2}}+20k^2$.
\end{proposition}
\begin{proof}
It is clear from the construction, that $s^{(k_1)}$ and $s^{(k_2)}$ will have the exact same behaviour for as long as there are at most $\min\{k_1,k_2\}$ miners in the mine. This property is true for $s^{(2)}$ and is preserved by the induction, since an epoch of $s^{(k_1)}$ and $s^{(k_2)}$ are defined similarly and all calls to lower-order strategies will also coincide. As a consequence, we have using the induction hypothesis that for all $k'\leq k-1$, and any $D\in \Nbb$,  $f_{s^{(k)}}(k',D) \leq c_{k'} D$.

We now turn to the analysis of $s^{(k)}$ when there are all $k$ miners in the mine. We remind the reader that due to the discussions in Section~\ref{sec:generaltm}, proving equation \eqref{eq:dynamic} is sufficient to conclude. We thus focus on showing that the cost of an epoch starting from a $(D,d)$-board and leading to a $(D',d')$-board satisfying \eqref{eq:split2} or \eqref{eq:join2} is bounded by, 
\begin{equation}
    f_{s^{(k-1)}}(k-1,D'-D)+kf_{s^{(k-1)}}(\ceil{k/2},\Delta)+10k^2\Delta
\end{equation}
The first term $f_{s^{(k-1)}}(k-1,D'-D)$ comes from the fact that only the last of unfinished bounded-horizon instances $s_{\ell}^{(k-1)}$ will have to deal with more than $\ceil{k/2}$ miners, and never with more than $k-1$ miners. Also, this term appears only if all other instances have reached their horizon and thus it defines the minimum depth $D'$ at the end of the epoch. The second term, $kf_{s^{(k-1)}}(\ceil{k/2},D-d)$ comes from the fact that there are at most $k-1$ other instances $s_{\ell}^{(k-1)}$ that are called, each involving at most $\ceil{k/2}$ miners. Finally, the term in $k^2\Delta$ decomposes as follows,
\begin{itemize}
    \item[--] initial non-lazy rebalancements (cost: $4k\Delta$) ;
    \item[--] rebalancements between instances (cost: $4k^2\Delta$) ;
    \item[--] provisioning for next epoch (cost: $2k\Delta$) ;
\end{itemize}
and is overall bounded by $10k^2\Delta$. The rest of the proof follows from Theorem \ref{th:mainth}. An illustration of the strategy defined here is provided in Figure \ref{fig:keq6}, for the specific case of $k=6$.
\end{proof}

\begin{figure}
\centering
\includegraphics[page=1,width=0.49\textwidth]{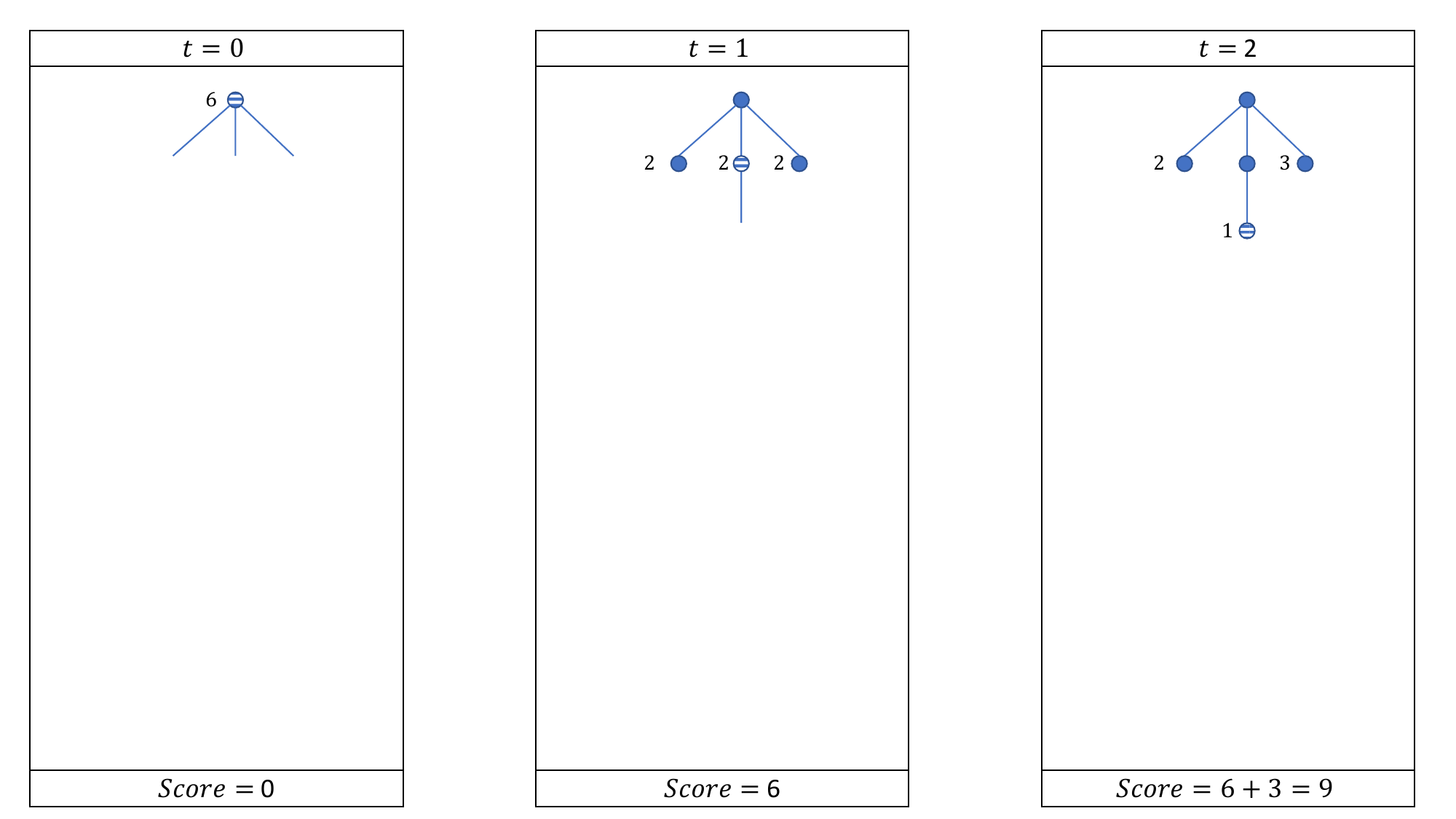}
\includegraphics[page=2,width=0.49\textwidth]{tree-mining-all.pdf}
\includegraphics[page=3,width=0.49\textwidth]{tree-mining-all.pdf}
\includegraphics[page=4,width=0.49\textwidth]{tree-mining-all.pdf}
\includegraphics[page=5,width=0.49\textwidth]{tree-mining-all.pdf}
\includegraphics[page=6,width=0.49\textwidth]{tree-mining-all.pdf}
\includegraphics[page=7,width=0.49\textwidth]{tree-mining-all.pdf}
\includegraphics[page=8,width=0.49\textwidth]{tree-mining-all.pdf}
\caption{Example of the course of the game with $k=6$ miners, for the tree-mining strategy defined above. At every instant, the node chosen by the adversary appears as dashed. In red is highlighted a move of type \eqref{eq:join} that leads to a change in the lowest common ancestor of all active leaves. }
\label{fig:keq6}
\end{figure}

\end{document}